\documentclass[letterpaper, 10 pt, conference]{ieeeconf}
\IEEEoverridecommandlockouts
\usepackage{booktabs}
\makeatletter
\let\NAT@parse\undefined
\makeatother
\usepackage{hyperref}
\usepackage{relsize}
\usepackage{amsmath,amssymb,amsfonts,mathtools}
\usepackage{nicematrix, nicefrac}
\usepackage{algorithmic}
\usepackage{textcomp}
\usepackage{xcolor}

\usepackage{booktabs}
\usepackage{graphicx}
\usepackage{epsfig} 
\usepackage{times}  
\usepackage{siunitx}
\NewCommandCopy\OldSI\SI
\RenewDocumentCommand\SI{ O{} m m }{  \nobreak  \OldSI[#1]{#2}{#3}  \nobreak}

\usepackage{tikz}
\usepackage{xstring}
\usepackage{forloop}
\usepackage{colortbl}
\usepackage{soul}
\usepackage{array}

\let\labelindent\relax  \usepackage{enumitem}

\usepackage{multirow} 
\usepackage{multicol}
\usepackage{arydshln}
\usepackage{pgfgantt}

\usepackage{textcomp}
\usepackage{rotating}
\usepackage{float}
\usepackage{algorithmic}
\usepackage{algorithm}
\usepackage[edges]{forest}

\usepackage[caption=false,font=footnotesize]{subfig}
\newcommand{\g}{\textnormal{\textsl{g}}}

\renewcommand{\c}{\mathrm{c}}
\renewcommand{\d}{\nabla}

\newcommand{\s}{{\mathrm{s}}}

\renewcommand{\c}{{\mathrm{c}}}
\renewcommand{\t}{{\mathrm{t}}}

\newcommand{\dotp}[2]{\langle #1,#2\rangle}

\newcommand{\cC}{{\mathcal C}}

\newcommand{\cG}{{\mathcal G}}

\newcommand{\cN}{{\mathcal N}}

\newcommand{\cS}{{\mathcal S}}

\newcommand{\cX}{{\mathcal X}}

\newcommand{\eqdef}{\mathrel{\mathop:}=}

\newcommand{\R}{\mathbb R}

\newcommand{\col}[1]{\mathrm{col}\!\left(#1\right)}

\makeatletter
  \renewcommand*\env@matrix[1][*\c@MaxMatrixCols c]{	\hskip -\arraycolsep
	\let\@ifnextchar\new@ifnextchar
	\array{#1}}
\newcommand*\bigcdot{\mathpalette\bigcdot@{1.5}}
\newcommand*\bigcdot@[2]{\mathbin{\vcenter{\hbox{\scalebox{#2}{$\m@th#1\cdot$}}}}}
\makeatother

\newcounter{excounter}[section]

\newcounter{rmcounter}

\newcounter{dfcounter}

\newcounter{thmcounter}

\newenvironment{thm}[1][\unskip]{	\refstepcounter{thmcounter}
	\vspace{10pt}
	\noindent \textbf{Theorem  \thethmcounter}~#1 \!\tikz\draw[black,fill=black] (0,0) circle (0.9pt); }
{		}

\newcounter{corcounter}[thmcounter]

\newcounter{prcounter}

\newcounter{clcounter}

\usepackage[dvipsnames]{xcolor}

\usepackage{cite}
\usepackage{amsthm}
\usepackage{mathrsfs} 
\usepackage{placeins}
\newtheorem{theorem}{Theorem}
\newtheorem{lemma}[theorem]{Lemma}

\theoremstyle{definition}

\newtheorem{assumption}{Assumption}
\newtheorem{problem}{Problem}

\theoremstyle{definition}
\newtheorem{remark}{Remark}

\usepackage{hyperref}
\hypersetup{
  colorlinks=true,
  linkcolor=blue,
  citecolor=blue,
  urlcolor=blue,
  pdfborder={0 0 0}
}

\makeatletter

\newcommand{\Rmnum}[1]{\expandafter\@slowromancap\romannumeral #1@}
\makeatother

\usepackage[caption=false,font=footnotesize]{subfig}

\title{\LARGE \bf
Quasi-Static Fault-Tolerant Feedback Control of a Quadrotor under Rotor Failure with Provable Safety Guarantees
}
\author{Mohamed Al Lawati$^{1}$ and Adeel Akhtar$^{2}$\thanks{$^{1}$ Department of Mechanical and Industrial Engineering, Sultan Qaboos University, Muscat, Oman. {\tt\small mlawati@squ.edu.om}}
\thanks{$^{2}$ Department of Mechanical and Industrial Engineering, New Jersey Institute of Technology, NJ, USA. {\tt\small adeel.akhtar@njit.edu}}
}

\begin{document}

\maketitle
\thispagestyle{empty}
\pagestyle{empty}

%%% Abstract %%%
\begin{abstract}
This paper presents a nonlinear control law for a quadrotor unmanned aerial vehicle (UAV) under single-rotor failure that guarantees set stabilization via quasi-static feedback (QSF).
Given a geometric curve in three-dimensional space, we characterize and stabilize the zero-dynamics manifold, also known as the path-following manifold, which represents all feasible motions along the path. Stabilizing this manifold ensures path-invariance: a UAV with a failed rotor initialized on the path with an appropriate orientation remains on the path for all future time. Furthermore, local exponential convergence to the manifold is guaranteed under certain conditions, implying that, under the stated assumptions, rotor failure during flight does not cause transverse deviation from the path. The proposed controller thus provides theoretical safety guarantees, which are validated through numerical experiments in the Drake physics-based simulation engine. The Code is publicly available at
\url{https://gradslab.github.io/quasistatic-ftc/}.
\end{abstract}
%%% Introduction %%%
\section{Introduction}
Unmanned aerial vehicles (UAVs) are increasingly deployed in safety-critical missions such as inspection, surveillance, and environmental monitoring, where sustained operation in cluttered environments is essential~\cite{ZinMenYouBin2025}.
In these settings, the UAV is required to accurately follow a prescribed spatial path rather than merely stabilize to a point or track a time–parameterized trajectory, since the mission objectives are typically defined geometrically~\cite{Akh22,AkhSalSha21,akhtar2020underactuated}. A fundamental challenge, however, is ensuring ``safety'' of UAVs in the case of actuator faults such as a complete failure of one of its rotors. In this work, we quantify safety in terms of path-invariance, which means, in simple terms, that once the UAV is on the path, it stays on the path for all future time despite the rotor failure. 

In mobile robotics, safety is typically enforced by constraining motion within a tube around a desired path~\cite{ames2016control} using control barrier and Lyapunov functions.
Path invariance is stricter, requiring the robot to remain exactly on the path, making control design more challenging---especially under rotor failure. 
This work designs a feedback controller guaranteeing that a UAV with one failed rotor converges to and remains on the desired path.

Fault-tolerant flight control (FTFC) for UAVs has been widely studied, covering both partial 
actuator degradation and complete rotor failure~\cite{NanSunScar2022}. Complete rotor 
loss is more challenging and practically relevant, as damaged rotors are typically disabled to avoid vibration. Early results established that rotor failure eliminates yaw controllability 
while altitude and position remain controllable~\cite{AkhWasNie2013,freddi_feedback_2011}.
Linear approaches~\cite{lippiello_emergency_2014-1} rely on linearization around a relaxed hovering equilibrium~\cite{lippiello_emergency_2014-1,mueller2016relaxed}.
Nonlinear methods that avoid linearization include robust feedback linearization~\cite{lanzon_flight_2014}, nonlinear dynamic inversion (NDI)~\cite{sun2021autonomous}, backstepping~\cite{lippiello_emergency_2014}, and incremental nonlinear dynamic inversion (INDI)~\cite{sun_incremental_2020}.   
Fault detection and isolation (FDI) remains an active research area~\cite{LuiGusJav2022}. However, the design of an FDI scheme is not addressed here; instead, we assume the availability of a reliable FDI module and an instantaneous transition of the system from the four-rotor configuration to the three-rotor mode upon fault occurrence.

In this work, we propose a quasi-static transverse feedback linearization (QSTFL) for a UAV with one failed rotor. In contrast to fault tolerant dynamic state-feedback linearization~\cite{AkhWasNie2013}, quasi-static feedback (QSF) achieves exact input–output linearization without additional controller states~\cite{LawLyn23,Rudolph2021,GstKolSch2024}, and requires inverting a $2\times2$ rather than a $3\times3$ decoupling matrix. To our knowledge, this is the first 
QSTFL-based UAV controller guaranteeing path invariance and local exponential convergence to 
all path motions under rotor failure. 
We make the following contributions:
\begin{enumerate}
    \item A quasi-static controller with closed-form expressions rendering the path exponentially 
    stable and forward invariant for a UAV with one failed rotor (Theorem~\ref{thm:main_result}).
    \item A diffeomorphic transformation converting the system into a partially linear system, and a proof showing boundedness of internal states.
    \item Validation in the Drake physics-based simulation engine, with code publicly available.
\end{enumerate}

The remainder of this paper presents the system model in Section~\ref{sec:modl}, formulates the problem in Section~\ref{sec:prblm}, develops the controller in Section~\ref{sec:fdk}, establishes its stability in Section~\ref{sec:stability_analysis}, and reports simulation results in Section~\ref{sec:simulations}, followed by concluding remarks in Section~\ref{sec:conclusion}.

\paragraph*{Notation}
The set of reals is $\R$, and $x\in\R^n$ is written as $\col{x_1,\cdots,x_n}$. The time derivatives are denoted $\dot{x}$, $\ddot{x}$, and $x^{(i)}$ for $i\geq 3$. For a matrix $R$, $R_i$ denotes its $i^{\rm th}$ column and $R_{ij}$ its $(i,j)^{\rm th}$ entry. The transpose of $x$ is $x^\top$, and $x\vert_{\cS}$ denotes the restriction of $x$ to a set $\cS$ with neighborhood $\cN_{\cS}$. The Euclidean norm is $\|x\|$, the inner product is $\langle x,y\rangle = x^\top y$, and the cross product is $x\times y$. For $f:\R^n\to\R$, the gradient is $\nabla f(x) = \frac{\partial f}{\partial x}(x)\in\R^n$ and the Jacobian is ${\rm d}f$. The unit sphere is $\mathbb{S}^n = \{x\in\R^{n+1}:\|x\|=1\}$, and the natural basis of $\R^3$ is $\{e_1,e_2,e_3\}$. 
Trigonometric functions $\sin\xi$, $\cos\xi$, and $\tan\xi$ are abbreviated as $\s_\xi$, $\c_\xi$, and $\t_\xi$, respectively.

%%% Modeling %%%
\section{Modeling} \label{sec:modl}

A typical quadrotor~\cite{LeeLeoMcc10}, hereafter referred to interchangeably as a UAV, moves in a navigational frame $\mathscr{N} = \{n_1,n_2,n_3\}$ aligned with north, east, and down, with origin $O_{\mathscr{N}}$. A body-fixed frame $\mathscr{B} = \{b_1,b_2,b_3\}$ is attached to the UAV with origin $O_{\mathscr{B}}$ at the center of mass (CoM), where $b_1$ points toward the heading and $b_3$ points down at hover. The CoM position and velocity relative to $\mathscr{N}$ are $p\in\mathbb{R}^3$ and $v=\dot{p}\in\mathbb{R}^3$. The attitude of $\mathscr{B}$ relative to $\mathscr{N}$ is represented by $R\in\mathrm{SO}(3)$ with columns $b_1,b_2,b_3$, and the body-frame angular velocity is $\Omega\in\mathbb{R}^3$. The UAV configuration evolves on a 6-dimensional manifold, while the full state used below is 12-dimensional. Moreover, it is actuated by thrust $u_t\in(0,\infty)$ along $-b_3$ and torque $\tau_a\in\mathbb{R}^3$, where the $i^{\rm th}$ torque component $\tau_{ai}$ acts about $b_i$, $i=1,2,3$. The system is thus underactuated with two degrees of underactuation. The UAV mass and moment of inertia are $m$ and $J$.

The equations of motion of a standard UAV~\cite{LeeLeoMcc10,HamUsmAkh25} are
\begin{subequations}
	\label{eq:UAV_geo}
	\begin{align}
		\dot{p} &= v, &
		\dot{v} &= \g e_3 -  \frac{1}{m}u_t R e_3 , \label{eq:vdt}  \\
		\dot{R} &= RS({\Omega}), &
		\dot{\Omega} &= J^{-1}(\tau_a - \Omega \times J\Omega).\label{eq:OMdt}
	\end{align}
\end{subequations}
The translational~\eqref{eq:vdt} and rotational~\eqref{eq:OMdt} dynamics couple through $Re_3 = b_3$. Among the twelve possible Euler parameterizations~\cite{SchJun18}, we choose the 2-1-3 configuration for its simple $b_3$ expression (no yaw dependence) and singularity-free hover. The rotational dynamics unidirectionally influence the translational dynamics, and as shown in 
Section~\ref{sec:fdk}, a damping term in $\tau_a$ is required for rotational stability. Thus, the applied torque, $\tau_a$, takes the form
\begin{equation}\label{eq:taua}
    \tau_a = \tau - J k_\Omega \Omega,
\end{equation}
where $k_\Omega > 0$ is a damping coefficient.

Let $\theta$, $\phi$, and $\psi$ be in $\R$ and represent respectively, pitch, roll, and yaw of the UAV. We define the UAV's attitude by $\eta := \col{\theta, \phi, \psi}$. 
Thus, {using the torque transformation~\eqref{eq:taua}, system} \eqref{eq:UAV_geo} is parameterized as 
\begin{equation}
	\label{eq:UAV_Euler}
	\begin{aligned}
		\dot{p} &= v, &
		\dot{v} &= \g e_3 - \frac{1}{m}u_t R_3, \\
		\dot{\eta} &= W\Omega, &
		\dot{\Omega} &=   J^{-1}(\tau - \Omega \times J\Omega) {-k_\Omega \Omega},
	\end{aligned}
\end{equation}
where
\begin{equation*}
	 R_3 = b_3 = \left[\begin{array}{ccc}\s_\theta\c_\phi\\ -\s_\phi\\ \c_\theta\c_\phi \end{array}\right],\text{and} \;\;
	W = \left[\begin{array}{ccc}
		\s_\psi / \c_\phi & \c_\psi / \c_\phi & 0 
		\\
		\c_\psi & -\s_\psi & 0 
		\\
		\s_\psi \t_\phi & \c_\psi \t_\phi & 1 
	\end{array}\right].
\end{equation*}
Next, we define the state and input vectors, respectively, as $x \eqdef \col{p,v,\eta,\Omega} \in  \R^{12}$ and ${U} \eqdef \col{u_t, \tau} \in (0,\infty) \times \R^3$. Hence, \eqref{eq:UAV_Euler} can be written in a control-affine form as
\begin{equation}
	\label{eq:UAV}
	\dot{x} = f(x) + {\cG}(x){U},
\end{equation}
where the vector fields $f:\R^{12} \to \R^{12}$ and $\cG:\R^{12} \to \R^{12}\times \R^{4}$ are given by \small
\begin{equation*}
	 f(x) = \!
	 \begin{bmatrix}
	 	v\\\g e_3 \\ W\Omega\\ -J^{-1}(\Omega \times J\Omega) - k_\Omega \Omega
	 \end{bmatrix}\!, 
	  \cG(x) = \!
	  \begin{bmatrix}
	  	0_{3\times 4}\\
	  	\begin{matrix}
	  		-Re_3/m & 0_{3\times 3}
	  	\end{matrix}\\
	  	0_{3\times 4}\\
	  	\begin{matrix}
	  		0_{3\times 1} & J^{-1}
	  	\end{matrix}
	  \end{bmatrix} \! .
\end{equation*}
\normalsize
The degree of underactuation of~\eqref{eq:UAV} is $2$. Also, the model~\eqref{eq:UAV} is not well-defined at the singularity point $\phi = \pm \pi/2$.

\subsection{UAV Rotor Failure Model}
Under single-rotor failure, it can be shown~\cite{HamUsmAkh25} that $\tau_1 = c_1 u_t - c_2\tau_3$, where $c_1=\ell/2$, $c_2=\ell/(2c_d)$, $\ell$ is the rotor-to-CoM distance, and $c_d$ is the propeller drag-to-thrust ratio. Since $\tau_1$ is no longer independent, the effective input reduces to $u=\mathrm{col}(u_t,\tau_2,\tau_3)$, and system~\eqref{eq:UAV} becomes
\begin{equation}\label{eq:UAV_fault}
    \dot{x} = f(x) + G(x)u,
\end{equation}
where $f$ is as in~\eqref{eq:UAV} and
\begin{equation*}
    G(x) = 
        \left[
        \begin{array}{cc}
        \multicolumn{2}{c}{0_{3\times 3}} \\[1pt]
        \hdashline\\[-10pt]
        -\frac{1}{m}Re_3 & 0_{3\times 2} \\[1pt]
        \hdashline\\[-10pt]
        \multicolumn{2}{c}{0_{3\times 3}} \\[1pt]
        \hdashline\\[-10pt]
        \multicolumn{2}{c}{J^{-1}C}
        \end{array}
        \right]\!\!, \text{where }  C =  \begin{bmatrix}
            c_1 & 0 & -c_2 \\
            0   & 1 & 0    \\
            0   & 0 & 1
        \end{bmatrix}.
\end{equation*}

%%% Problem formulation %%%
\section{Problem formulation}\label{sec:prblm}

The objective is to drive the UAV CoM to a geometric path $\mathcal{C} \subset \mathbb{R}^3$, where geometric means $\mathcal{C}$ is a curve in space not parameterized by time, while enforcing invariance of motions along $\mathcal{C}$ despite complete single-rotor failure. Invariance guarantees that a UAV on the path remains on it for all future time regardless of rotor failure. Such a property cannot be guaranteed by conventional trajectory-tracking controllers. A curve $\mathcal{C}\subset\mathbb{R}^3$ is represented as the intersection of two independent surfaces:
\begin{equation}\label{eq:C}
\mathcal{C} = \{p\in\mathbb{R}^3: h_1(p)=h_2(p)=0,\; \d h_1(p)\times \d h_2(p)\neq 0\}.
\end{equation}
As in~\cite{AkhWasNie2013}, $\mathcal{C}$ is not required to be closed, thereby relaxing the closedness assumption in~\cite{RozMag2012}. The three-dimensional input $u$ provides sufficient control authority to enforce the constraints $h_1(p)=h_2(p)=0$. We call them transverse constraints. The additional degree of actuation of $u$ is used to define motion along $\cC$. Thus, we define a third output $h_3(p)$ to be a path-coordinate function, e.g., the arc length along $\mathcal{C}$ from a reference point. Unlike $h_1$ and $h_2$, the output $h_3$ is not necessarily regulated to zero; instead, it is used to prescribe the desired motion along $\mathcal{C}$. The output map is
\begin{equation}\label{eq:output}
h(p)=\mathrm{col}\bigl(h_1(p),h_2(p),h_3(p)\bigr).
\end{equation}

\begin{assumption}\label{asp:pp}
    Given a UAV and a path $\cC$, there exists a neighborhood $\cN_{\cC}$ of    $\cC$, such that the following conditions are satisfied on $\cN_{\cC}$
    \begin{enumerate}[label=\textbf{A\arabic*} ]
        \item \label{asp:ut}$u_t \neq 0$ and
        \item \label{asp:dt} $\dotp{R_3}{\d h_2} \neq 0$.
    \end{enumerate}
\end{assumption}

Assumption~\ref{asp:ut} requires nonzero thrust, as $u_t = 0$ implies 
all motors off. For Assumption~\ref{asp:dt}, $\langle \nabla h_2, R_3 
\rangle = 0$ on $\mathcal{C}$ would mean the thrust vector direction,
$-R_3$, lies in the tangent plane of the surface $h_2 = 0$, 
which is physically impossible for horizontal paths, and 
requires $u_t < 0$ for non-horizontal ones. Hence, $\langle 
\nabla h_2, R_3 \rangle \neq 0$ on $\mathcal{C}$ is a mild 
and practically feasible assumption. Henceforth, we assume our states evolve on the regular domain $\cX :=
    \bigl\{x\in\mathcal \R^{12}:\;
    u_t(x)\neq 0,\;
    \dotp{\d h_2}{R_3}\neq 0,\;
    \phi\neq \pm \pi/2
    \bigr\}$.

\begin{problem}[Fault-tolerant Path Following Problem (Ft-PFP)]
\label{prob:Quasi-PF-Problem} 
Given {a} UAV {with a completely failed rotor}~\eqref{eq:UAV_fault} and a path~\eqref{eq:C} satisfying Assumption~\ref{asp:pp}, design feedback controller $\kappa: \R^{12} \to \R^3$, $x\mapsto \kappa(x) = u$, such that the closed-loop UAV system achieves the following goals:
\begin{enumerate}[label=\textbf{G\arabic*} ]
	\item \label{goal:1} The UAV CoM approaches $\cC$ exponentially.
	\item \label{goal:2} The UAV CoM's motion along $\cC$ is controlled-invariant, i.e., if the UAV CoM is initialized on $\cC$ with its velocity vector pointing along the tangent vector to $\cC$, then $\kappa(x)$ ensures the UAV CoM never leaves $\cC$.
    \item \label{goal:3} The UAV CoM follows a prescribed tangential motion along $\mathcal C$.
    \item \label{goal:4} The UAV yaw motion remains bounded.
	 \end{enumerate}
\end{problem}
\begin{remark}\label{rmk:dea}
The authors in~\cite{AkhWasNie2013} and~\cite{HamUsmAkh25} solved a similar Ft-PFP using dynamic extension, introducing two additional controller states (typically thrust and its time-derivative) and a new input (typically thrust's second time-derivative) to feedback-linearize the extended dynamics. Their approach has two limitations. First, recovering the thrust requires either solving a differential equation involving the thrust or double-integrating a new input defined in terms of the thrust and its first time derivative. Hence, sensor measurement is required for thrust and its time-derivative, posing practical challenges. Second, gain tuning involves two extra states. The proposed approach overcomes both limitations.
\end{remark}

%%% Solution using a QSF %%%
\section{Solution using a QSF}\label{sec:fdk}

Our approach separates the output coordinates into transverse and tangential components. The transverse outputs $h_1$ and $h_2$ define the desired path $\mathcal{C}$ and are used to enforce convergence to and invariance of the path. The third output $h_3=s(p)$ parameterizes motion along $\mathcal{C}$ and is not, in general, regulated to zero. Instead, it is used to assign the desired tangential motion along the path, such as point stabilization, velocity tracking, or acceleration tracking. Thus, the QSF design stabilizes the transverse dynamics while leaving the tangential dynamics assignable through a user-specified longitudinal command. 

Consider the system~\eqref{eq:UAV_fault} and the output~\eqref{eq:output}. We differentiate each output component until an input appears. This yields $\col{h_1^{(r_1)}, h_2^{(r_2)}, h_3^{(r_3)}} = Au + \col{L_f^{r_1}h_1, L_f^{r_2}h_2, L_f^{r_3}h_3}$,
where $A \in \R^{3\times 3}$ is the decoupling matrix whose $(i,j)^{\text{th}}$ entry is $L_{g_j}L_f^{r_i-1}h_i$, and $r_i$ is the lowest derivative order of output $h_i$ such that at least one input shows up \cite{Isi95}. Invertibility of the decoupling matrix $A$ determines the feasibility of input-output linearization. For our system, we have $r_1=r_2=r_3=2$. Since the failed-rotor system has the input vector $u=\mathrm{col}(u_t,\tau_2,\tau_3)$, the $i^{\text{th}}$ row of $A$ is $-\frac{1}{m}\left[\langle R_3,\d h_i\rangle,\;0,\;0\right], i=1,2,3$. Notice that $A$ is singular of constant rank $1$ whenever at least one of the scalars $\langle R_3,\d h_i\rangle$ is nonzero. On $\cX$, we have $\langle R_3,\d h_2\rangle\neq 0$, so the row corresponding to $h_2$ provides one independent input-output relation. This means that we have a single  independent relation between the inputs and  output derivatives. Hence, we assign the auxiliary input $\nu_2$ as
\begin{equation} \label{eq:nu2}
        \nu_2 := \ddot{h}_2 =-\frac{1}{m} \dotp{R_3}{\d h_2}u_t + \beta_2, 
\end{equation}
where $\beta_2 = v^\top H_2v + \dotp{\d h_2}{\g e_3}$, and $H_2$ is the Hessian of $h_2$ with respect to $p$. We need to design $\nu_2$ to send $h_2$ to zero. To do so, let $Z^2:\mathcal X\to\mathbb R^2$ be the map defined by $Z^2(x):=
\col{
h_2(p),\,
\dot h_2(x)}.$ Define the transformed state $z^2$ as $z^2 := Z^2(x)$. Hence, using \eqref{eq:nu2}, the $z^2$-dynamics is 
\begin{equation}\label{eq:z2}
    \dot{z}^2 = A_{c2}z^2 + B_{c2}\nu_2,
\end{equation}
where $(A_{c2}$, $B_{c2})$ is a 2-by-2 Brunovsky pair~\cite{MarTom98}.
We choose
\begin{equation}\label{eq:nui2}
    \nu_2 = -K_2 z^2,
\end{equation}
where $K_2 \in \R^{1\times 2}$ such that $A_{c2} - B_{c2}K_2$ is Hurwitz. This choice stabilizes the origin of the $z^2$-dynamics. Solving~\eqref{eq:nu2} for $u_t$, one obtains an algebraic expression for the thrust as
\begin{equation}\label{eq:ut}
    u_t = m\frac{\beta_2-\nu_2}{\dotp{\d h_2}{R_3}},
\end{equation}
where $\nu_2$ is given in~\eqref{eq:nui2}, and must be designed such that $\nu_2 \neq \beta_2$.

\begin{remark}
    In this design, we compute the thrust algebraically by performing elementary algebraic operations. This is unlike the dynamic controller of~\cite{AkhWasNie2013}. See Remark~\ref{rmk:dea} for the benefits that a QSF offers.
\end{remark}

To acquire the required two equations for the remaining inputs $\tau_2$ and $\tau_3$, we take extra time derivatives of $\ddot{h}_1$ and $\ddot{h}_3$. It turns out that $h_1^{(4)}$ and $h_3^{(4)}$ provide two additional independent equations in terms of $\tau_2$ and $\tau_3$. In particular,  we have
\begin{equation}   \label{eq:nu13} 
       \begin{bmatrix} h_1^{(4)} \\ h_3^{(4)}\end{bmatrix} =  
        \tilde{A}\begin{bmatrix}
                \tau_2 \\ \tau_3
            \end{bmatrix}  +\tilde{\beta}(x),
    \end{equation}
where $\tilde{\beta}_i = L_f^4 h_i(x)$ whose explicit expressions are omitted for brevity, and
\begin{equation}\label{eq:Atilde}
    \begin{aligned}
        \tilde{A} =
        \frac{u_t}{m\dotp{\d h_2}{R_3}}\begin{bmatrix}
            (\d h_1 \times \d h_2)^\top \\
            (\d h_3 \times \d h_2)^\top
        \end{bmatrix}
        \begin{bmatrix}
            \frac{ R_2}{J_{2}} & \frac{-c_2R_1}{J_{1}}
        \end{bmatrix}.
    \end{aligned}
\end{equation}
\begin{remark}
    This decoupling matrix $\tilde{A}$ is $2$-by-$2$ as compared to the 3-by-3 decoupling matrix of~\cite{AkhWasNie2013}. This dimension simplification is a result of pre-computing the thrust in~\eqref{eq:ut}.
\end{remark}
We define the auxiliary inputs
\begin{equation}\label{eq:nu13_def}
    \begin{bmatrix}\nu_1 \\ \nu_3\end{bmatrix}:=\begin{bmatrix} h_1^{(4)} \\ h_3^{(4)}\end{bmatrix} =  
        \tilde{A}\begin{bmatrix}
                \tau_2 \\ \tau_3
            \end{bmatrix}  +\tilde{\beta}(x),
\end{equation}
We need to design $\nu_1$ to send $h_1$ to zero. The design of $\nu_3$ should allow for a desired motion along the path. To do so,
let $Z^1:\mathcal X\to\mathbb R^4$ denote the transverse coordinate map defined as $Z^1(x):=
\col{
h_1(p),\,
\dot h_1(x),\,
\ddot h_1(x),\,
h_1^{(3)}(x)
}$.
Using the map $Z^1$, we define transformed state $z^1$ as
$
    z^1 := Z^1(x) = 
    \col{
        h_1(x),\,
        \dot{h}_1(x),\,
        \ddot{h}_1(x),\,
        h_1^{(3)}(x)
        }.
$
Then
\begin{equation}\label{eq:z1}
    \dot{z}^1 = A_{c4}z^1 + B_{c4}\nu_1,
\end{equation}
where $(A_{c4},B_{c4})$ is the $4$-by-$4$ Brunovsky pair. The input
\begin{equation}\label{eq:nu1}
    \nu_1 = -K_1 z^1
\end{equation}
is chosen such that $A_{c4}-B_{c4}K_1$ is Hurwitz; hence it stabilizes the origin of the $z^1$-dynamics. 
Together with $z^2$ in~\eqref{eq:z2}, define the transverse state
$
    z^\perp := \mathrm{col}(z^1,z^2) = \col{Z^1(x), Z^2(x)} \in \R^6.
$
Equations~\eqref{eq:z2} and~\eqref{eq:z1} define the transversal dynamics driven by the auxiliary inputs $\nu_2$ and $\nu_1$, respectively. We underscore that, uppercase letters such as $Z^i$ denote coordinate maps on $\mathcal X$, while lowercase letters denote their values along trajectories; that is, $z^i=Z^i(x)$.

For the tangential output $h_3=s(p)$, define
$
    z^\parallel := z^3 :=
    \mathrm{col}\bigl(
        h_3(x),\,
        \dot{h}_3(x),\,
        \ddot{h}_3(x),\,
        h_3^{(3)}(x)
    \bigr) \in \R^4.
$
The corresponding tangential dynamics are
\begin{equation}\label{eq:zparallel}
    \dot{z}^\parallel = A_{c4}z^\parallel + B_{c4}\nu_3.
\end{equation}

\begin{remark}\label{rmk:tangential_control}
Unlike $z^\perp$, the state $z^\parallel$ is not necessarily stabilized to the origin of the $z^{\parallel}$-dynamics. Instead, $\nu_3$ is used to assign the desired behavior of $z^{\parallel}$ along the path $\mathcal{C}$. For example, to track a desired path speed $v_d(t)$, define $ e_v:=\mathrm{col}\bigl(\dot h_3-v_d,\ddot h_3-\dot v_d,h_3^{(3)}-\ddot v_d\bigr)$, and choose $\nu_3=v_d^{(3)}-K_3e_v$ with a stabilizing $K_3$. Point stabilization or path-coordinate trajectory tracking are obtained similarly by defining the corresponding error in $h_3$ and its derivatives.
\end{remark}

The torque inputs $\tau_2, \tau_3$ can easily be computed using~\eqref{eq:nu13_def}. Note that
$        \det( \tilde{A} ) =
       -c_2 \frac{   \dotp{\d h_1 \times \d h_2}{\d h_3}   u_t^2}{m^2J_{1}J_{2}\dotp{\d h_2}{R_3}}$. 
On $\cX$, we have  $u_t\neq 0$ and $\dotp{\d h_2}{R_3}\neq 0$. Also, the surfaces $h_1$ and $h_2$ are independent by definition of $\mathcal{C}$, implying $\d h_1\times \d h_2\neq 0$ in a neighborhood of $\mathcal{C}$. Since $h_3=s(p)$ is the path-coordinate function, its gradient has a nonzero projection along the tangent direction of $\mathcal{C}$. Equivalently, $\dotp{\d h_1\times \d h_2}{\d h_3}\neq 0 $ on $\cX$. Therefore, $\det(\tilde{A})\neq 0$ whenever the closed-loop trajectory remains in this regular domain. The parameters $m$, $J_{1}$, $J_{2}$, and $c_2$ are positive scalars. 
The torques $\tau_2$, $\tau_3$ are therefore given by
\begin{equation}\label{eq:tau}
 \begin{bmatrix}\tau_2\\ \tau_3\end{bmatrix} = \tilde{A}^{-1}\left(
 \begin{bmatrix}
        \nu_1\\ \nu_3 
    \end{bmatrix}
    - \begin{bmatrix}
        \tilde{\beta}_1(x)\\
        \tilde{\beta}_3(x)
        \end{bmatrix}
        \right),
\end{equation}
where $\nu_1$ and $\nu_2$ are chosen by~\eqref{eq:nu1} and~\eqref{eq:nui2} while $\nu_3$ assigns the tangential motion as described in Remark~\ref{rmk:tangential_control}. As a result, the complete QSF is given by~\eqref{eq:ut} and \eqref{eq:tau}.
%%% Stability analysis %%%
\section{Stability analysis}
\label{sec:stability_analysis}

Recall the transverse coordinate maps $Z^1:\mathcal X\to\R^4$ and $Z^2:\mathcal X\to\R^2$ defined in Section~\ref{sec:fdk}. Let $
Z^\perp(x):=\mathrm{col}\bigl(Z^1(x),Z^2(x)\bigr).
$
The path-following manifold is
\begin{equation}
\label{eq:x-star}
    \mathcal X^\ast
    :=
    \{x\in\mathcal X:Z^\perp(x)=0\}.
\end{equation}
Equivalently, $\mathcal{X}^\ast=\{x\in\mathcal{X}:z^1=0,\ z^2=0\}$. 
The set $\mathcal{X}^\ast$ represents motions for which the UAV CoM lies on $\mathcal{C}$ with the required transverse derivative conditions. The tangential state $z^\parallel$ is not required to vanish on $\mathcal{X}^\ast$; instead, it determines the motion along $\mathcal{C}$.

\begin{figure*}[!htb]
    \centering
        \subfloat[Stage 1:  $t=\SI{0}{\second}$\label{fig:sub1}]{
        \includegraphics[width=0.23\textwidth]{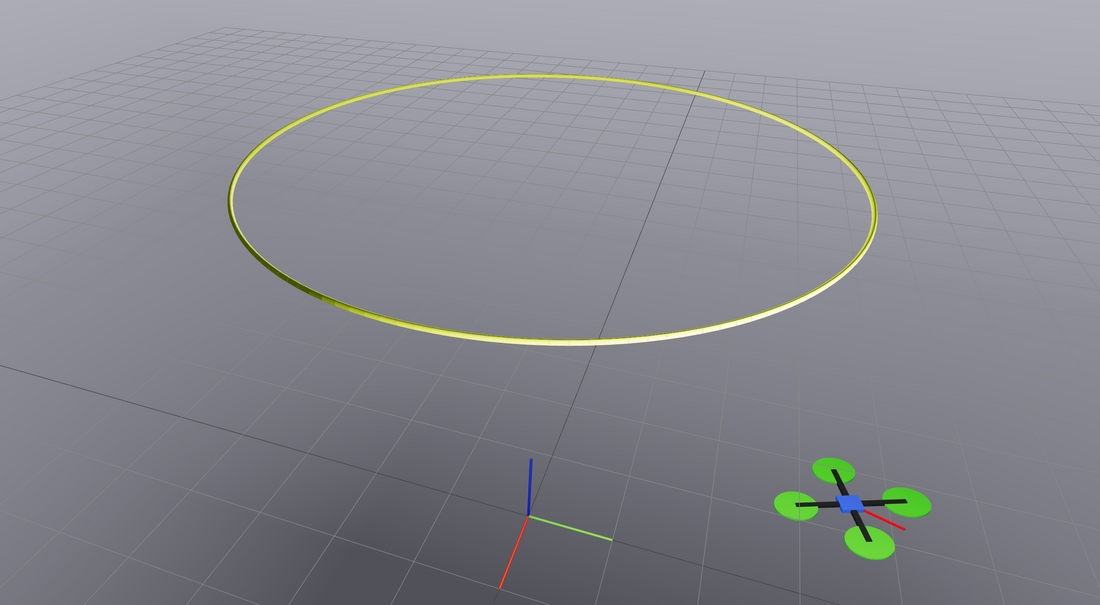}
    }
    \subfloat[Stage 1:  $t=\SI{2}{\second}$\label{fig:sub2}]{
        \includegraphics[width=0.23\textwidth]{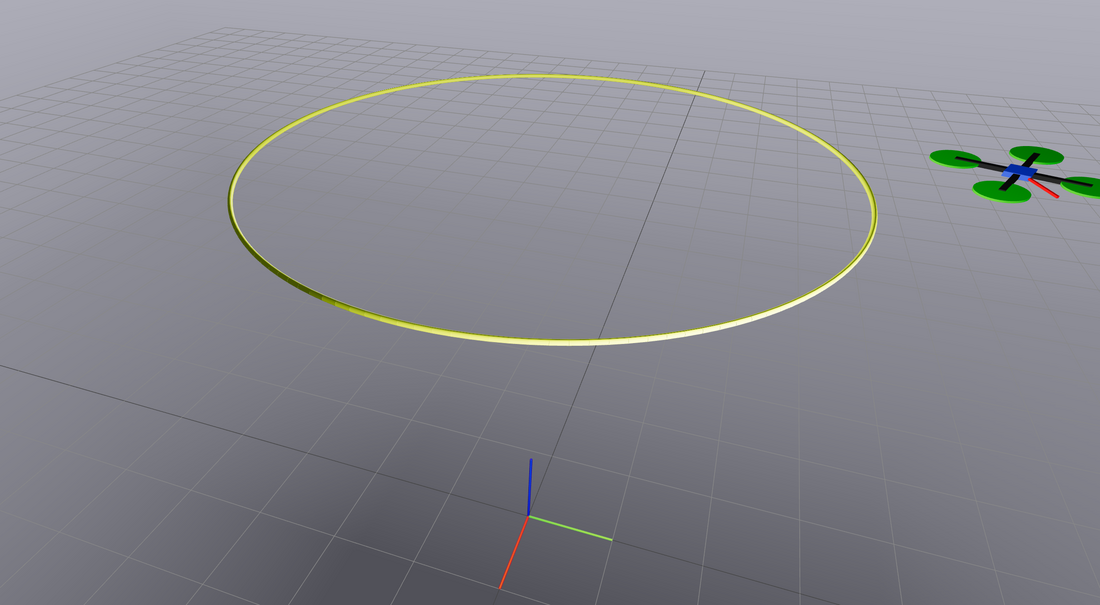}
    }
    \subfloat[Stage 1:  $t=\SI{4}{\second}$\label{fig:sub3}]{
        \includegraphics[width=0.23\textwidth]{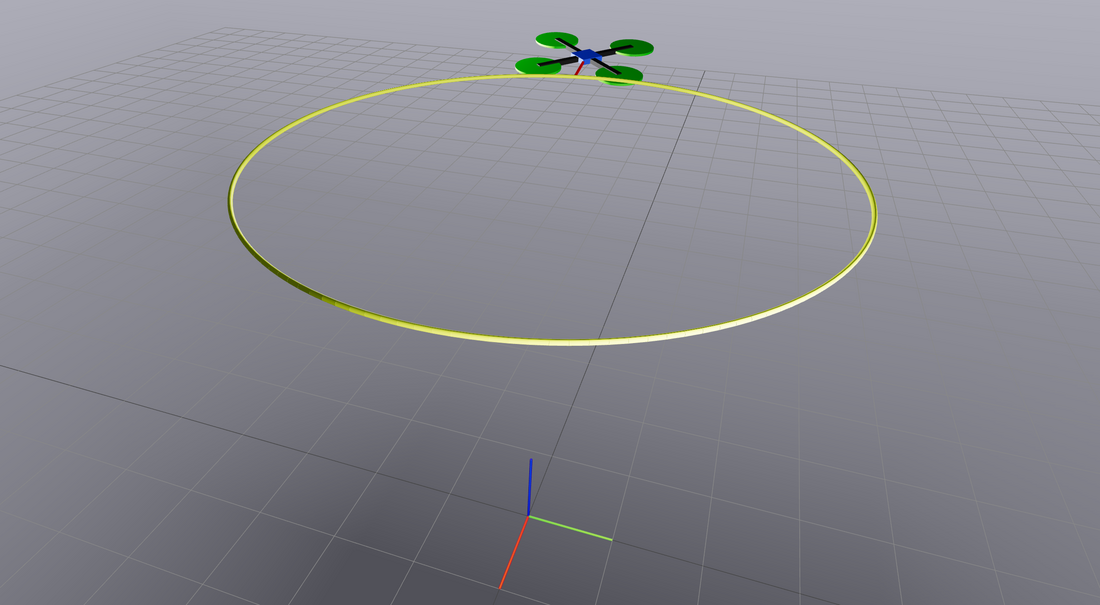}
    }
    \subfloat[Stage 1:  $t=\SI{5}{\second}$\label{fig:sub4}]{
        \includegraphics[width=0.23\textwidth]{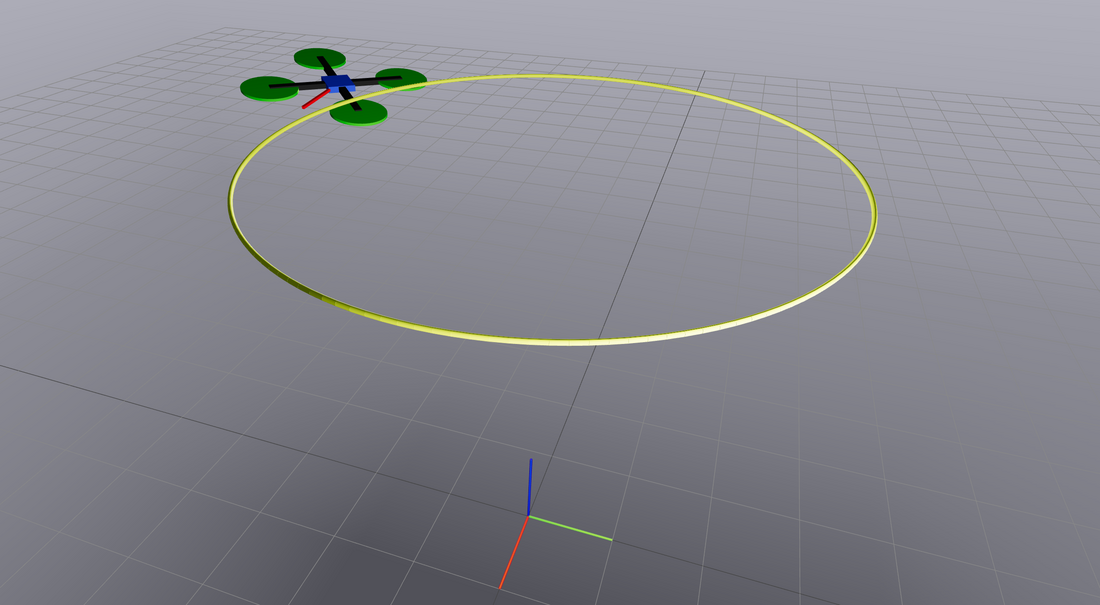}
    }\\[1ex] \vspace{-10pt}
        \subfloat[Stage 2:   $t=\SI{10}{\second}$\label{fig:sub5}]{
        \includegraphics[width=0.23\textwidth]{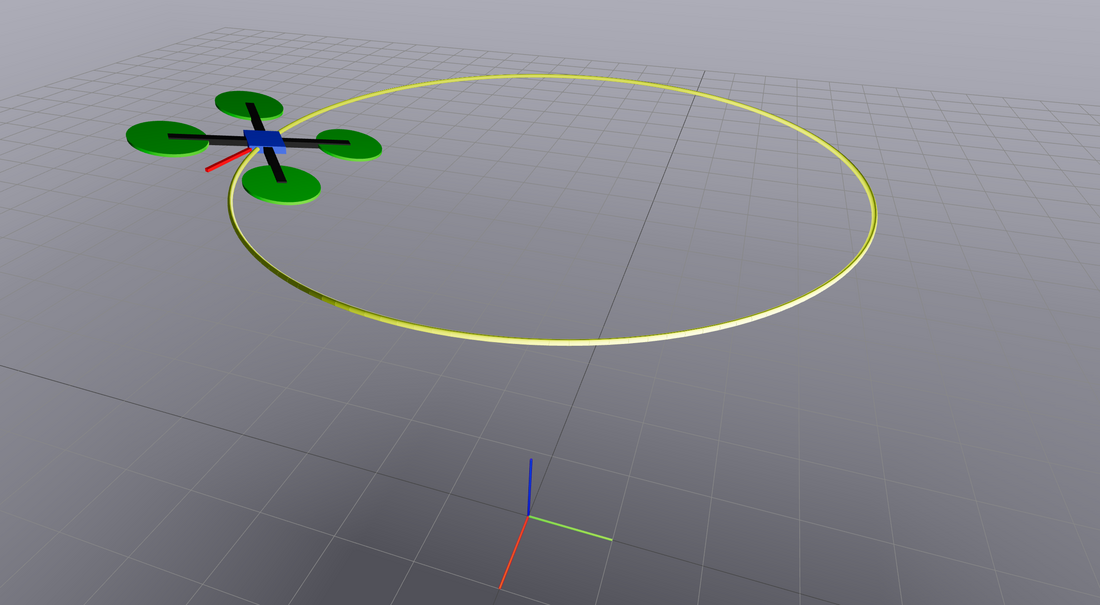}
    }
    \subfloat[Stage 2:  $t=\SI{11}{\second}$\label{fig:sub6}]{
        \includegraphics[width=0.23\textwidth]{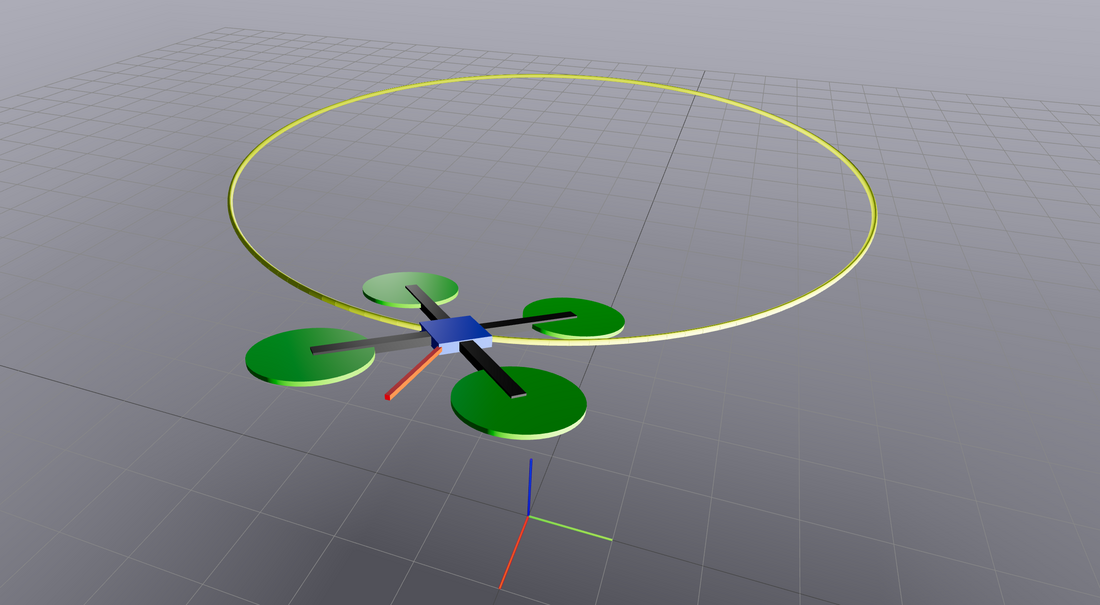}
    }
    \subfloat[Stage 2: $t=\SI{12}{\second}$\label{fig:sub7}]{
        \includegraphics[width=0.23\textwidth]{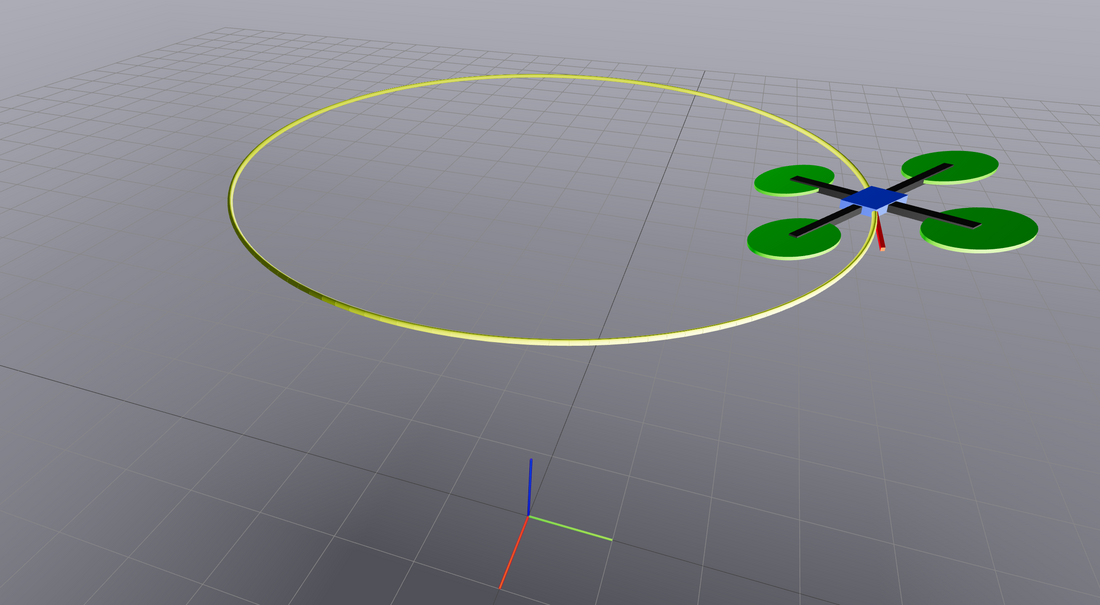}
    }
    \subfloat[Stage 2:  $t=\SI{13}{\second}$\label{fig:sub8}]{
        \includegraphics[width=0.23\textwidth]{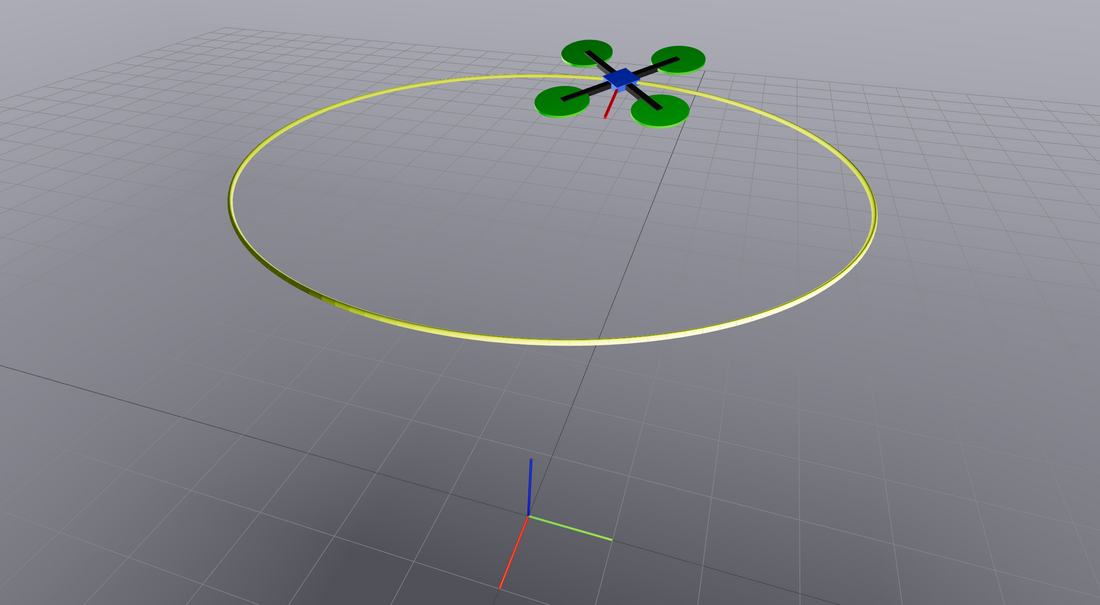}
    }\\[1ex] \vspace{-10pt}
        \subfloat[Stage 3:  $t=\SI{20}{\second}$\label{fig:sub9}]{
        \includegraphics[width=0.23\textwidth]{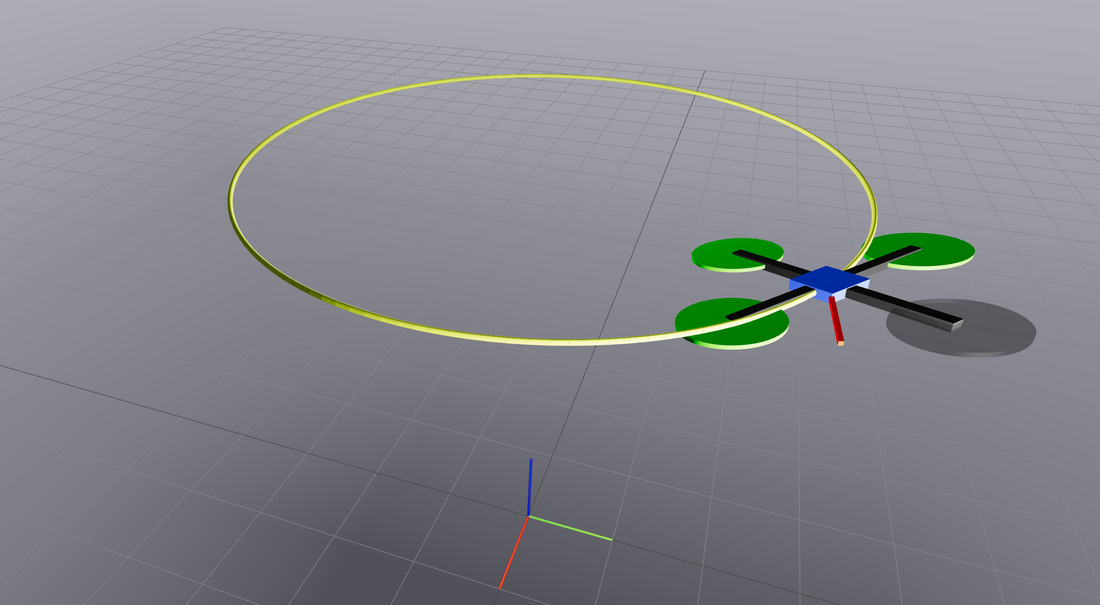}
    }
    \subfloat[Stage 3:  $t=\SI{21}{\second}$\label{fig:sub10}]{
        \includegraphics[width=0.23\textwidth]{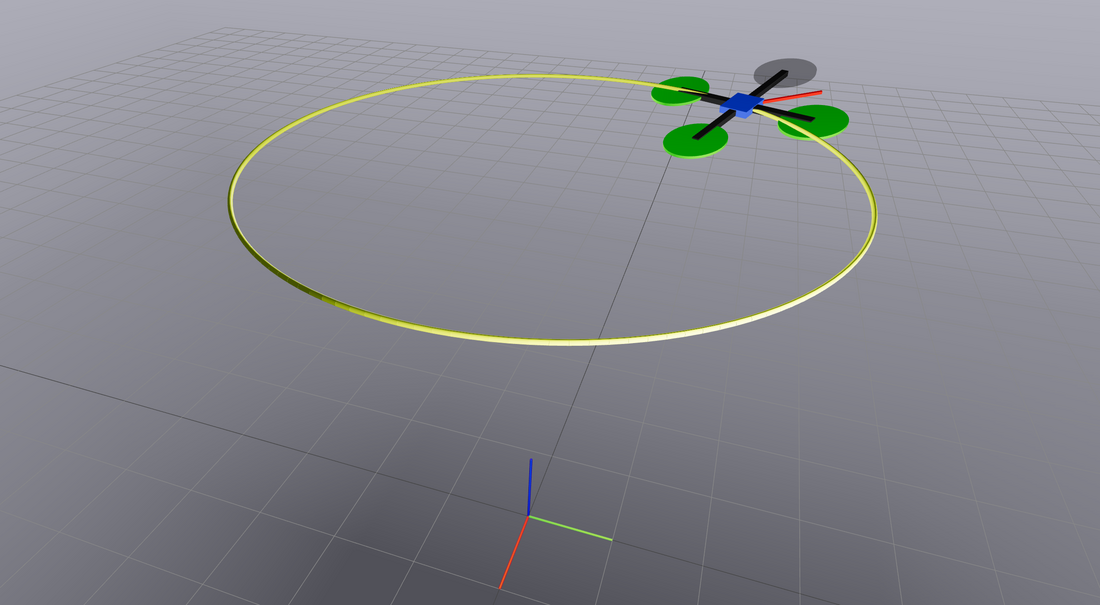}
    }
    \subfloat[Stage 3: $t=\SI{22}{\second}$\label{fig:sub11}]{
        \includegraphics[width=0.23\textwidth]{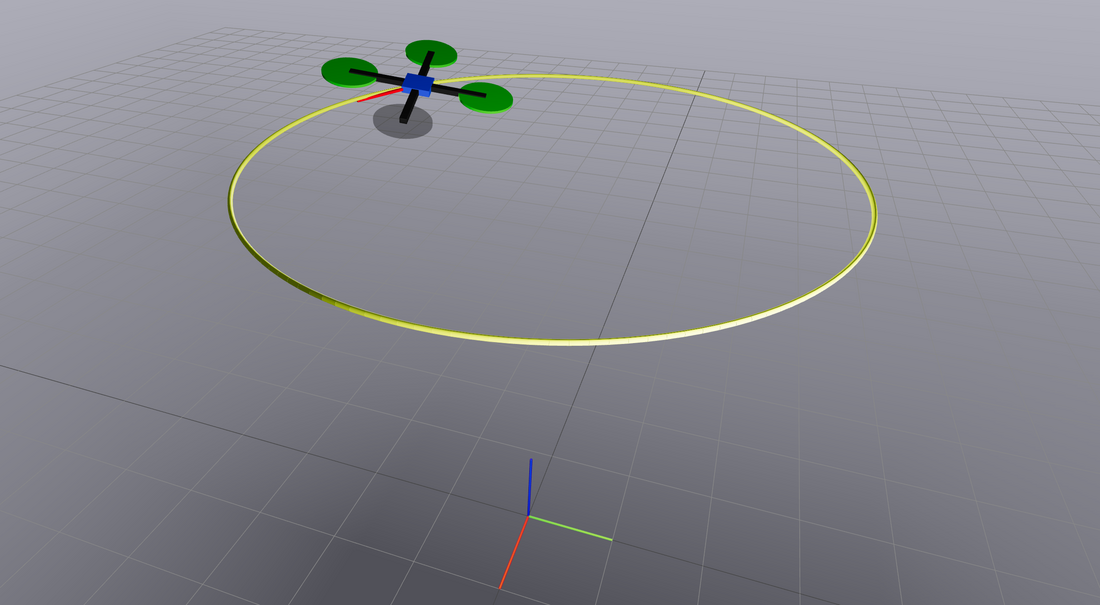}
    }
    \subfloat[Stage 3:  $t=\SI{23}{\second}$\label{fig:sub12}]{
        \includegraphics[width=0.23\textwidth]{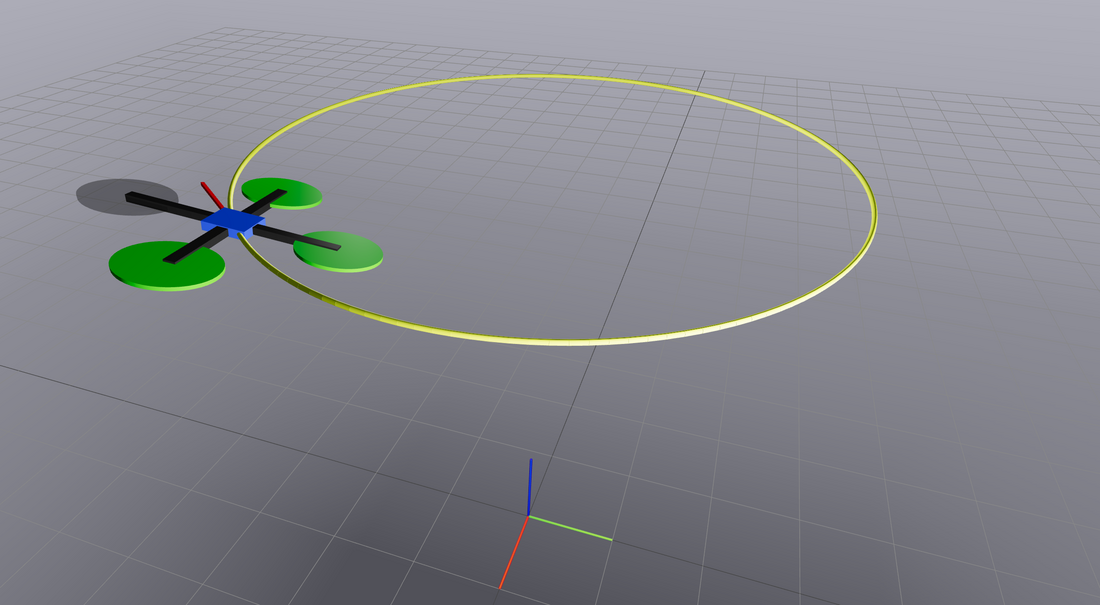}
    }

    \caption{Snapshots of the numerical simulation performed in Drake highlighting motion invariance despite sudden rotor failure.}
    \label{fig:snapshots}
\end{figure*}
On $\cX$, define a mapping $    \sigma:\mathcal{X}\to \R^{12}, x\mapsto (z^\perp,z^\parallel,\xi)$, where $\xi\in \mathbb{S}^1\times\R$ is chosen so that $\sigma$ defines a local coordinate transformation. It suffices to choose $\xi$ such that the Jacobian of $\sigma$, denoted by ${\rm d}\sigma$, is nonsingular. We choose
\begin{equation}\label{eq:xi}
    \xi =
    \mathrm{col}\bigl(
        \psi,\,
        J_1\Omega_1+J_3c_2\Omega_3
    \bigr)
    \in \mathbb{S}^1\times\R .
\end{equation}
The determinant of the Jacobian of $\sigma$ is
\begin{equation*}
    \det({\rm d} \sigma) = -c_2J_3\frac{(\beta_2 - \nu_2)^4}{\dotp{R_3}{\d h_2}^6} \dotp{\d h_1 \times \d h_2}{\d h_3}^4 \cos(\phi).
\end{equation*}
On $\cX$, $\det({\rm d} \sigma)$ is well defined implying the Jacobian ${\rm d}\sigma$ is nonsingular. Therefore, $\sigma$ is a local diffeomorphism. In the transformed coordinates, the internal dynamics are given by
\begin{align}\label{eq:intr_dyn}
    \dot{\xi} = 
    \left.\begin{bmatrix}
         e_3^\top W(\eta)\Omega\\
        -[1, 0, c_2] (\Omega \times J\Omega  +  k_\Omega J\Omega ) + c_1 u_t 
    \end{bmatrix}\right|_{x=\sigma^{-1}(z^\perp,z^\parallel,\xi)}.
\end{align}
Note that although $\sigma^{-1}$ might be difficult to compute, it is guaranteed to exist since $\sigma$ is a diffeomorphism.

The coordinate transformation $\sigma$ maps $x$ to $(z^\perp,z^\parallel,\xi)$. The QSF controller~\eqref{eq:ut},~\eqref{eq:tau}, together with the auxiliary inputs $\nu_1=-K_1z^1$ and $\nu_2=-K_2z^2$, stabilizes the transverse state $z^\perp$. The tangential state $z^\parallel$ evolves according to the selected input $\nu_3$ and is used to prescribe the motion along the path. Therefore, the stability analysis below concerns the path-following manifold $\mathcal{X}^\ast=\{x:Z^\perp(x)=0\}$ and the boundedness of the internal state $\xi$. 

The next lemma proves the boundedness of the angular velocity states. This fact will be used in our main stability result.
\begin{lemma} \label{lma:OM_bound}
    Let 
    \begin{equation} \label{eq:rot_dyn}
        \dot{\Omega} = J^{-1}(\tau - \Omega \times J\Omega) - k_\Omega \Omega,
    \end{equation}
    where $\Omega \in \R^3$ is the state and $\tau \in \R^3$ is the input. Then, the dynamics~\eqref{eq:rot_dyn} is input-to-state stable.
\end{lemma}
\begin{proof}
    From~\cite[LemmaVI.1]{AkhWasNie2013}, the origin of the unforced case of~\eqref{eq:rot_dyn}, i.e. $\tau=0$, is globally exponentially stable. From~\cite[Lemma 4.6]{Kha02}, the dynamics~\eqref{eq:rot_dyn} is input-to-state stable.
\end{proof}
Under the QSF law \eqref{eq:ut} and \eqref{eq:tau}, the torque input $\tau$ 
in~\eqref{eq:rot_dyn} is bounded since $\nu_1$, $\nu_3$, 
and $z^\perp$ are bounded. Hence, by Lemma~\ref{lma:OM_bound}, $\Omega$ 
remains bounded.

\begin{thm}
\label{thm:main_result}
Consider the UAV model with a completely failed rotor in~\eqref{eq:UAV_fault} where $x \in \cX$. Let the path $\mathcal C$ be described as in~\eqref{eq:C}. Under the QSF law~\eqref{eq:ut} and~\eqref{eq:tau}, with $\nu_1=-K_1z^1$ and $\nu_2=-K_2z^2$, the path-following manifold $\mathcal X^\ast$ defined in~\eqref{eq:x-star} is locally exponentially stable and forward invariant. Moreover, if the tangential input $\nu_3$ is locally bounded, then the internal state $\xi$ and its time derivative $\dot{\xi}$ remain bounded.
\end{thm}
\begin{proof}
Under the auxiliary inputs $\nu_1=-K_1z^1$ and $\nu_2=-K_2z^2$, the transverse dynamics are
\begin{equation}\label{eq:zperp_cl}
    \dot{z}^\perp
    =
    \begin{bmatrix}
        A_{c4}-B_{c4}K_1 & 0 \\
        0 & A_{c2}-B_{c2}K_2
    \end{bmatrix}
    z^\perp .
\end{equation}
Since both diagonal blocks are Hurwitz, the origin $z^\perp=0$ is exponentially stable. Hence, the path-following manifold
$
    \mathcal{X}^\ast
    =
    \{x\in\mathcal{X}:Z^\perp(x)=0\}
$
is locally exponentially stable in the transformed coordinates. Moreover, if $Z^\perp(0)=0$, then $z^\perp(t)=0$ for all $t\geq 0$, which establishes forward invariance of $\mathcal{X}^\ast$.
Since $\sigma$ is a local diffeomorphism on $\mathcal{X}$, local exponential stability of $z^{\perp} = 0$ in the transformed coordinates implies local exponential stability of $\mathcal{X}^*$ in 
the original coordinates $x$.
It remains to show boundedness of $\xi$ and $\dot{\xi}$. Since the result is local on the regular domain $\mathcal X$, we restrict attention to the neighborhood of $\mathcal X^\ast$ on which the coordinate transformation $\sigma$ is a diffeomorphism and the QSF law is well defined. In this neighborhood, the quantities appearing in the denominators of the controller and of ${\rm d}\sigma$, namely $u_t$, $\langle \d h_2,R_3\rangle$, $\cos\phi$, and $\langle \d h_1\times \d h_2,\d h_3\rangle$, are bounded away from zero. Therefore, by continuity of the QSF law and local boundedness of $\nu_3$, the resulting inputs $u_t$ and $\tau$ are locally bounded along closed-loop trajectories. By Lemma~\ref{lma:OM_bound}, boundedness of $\tau$ implies that $\Omega$ remains bounded.
The first component $\xi_1=\psi$ belongs to $\mathbb S^1$ and is therefore bounded modulo $2\pi$. 
Moreover, since $\cos\phi$ remains nonzero in the regular neighborhood, possibly after shrinking this neighborhood there exists $M_\phi>0$ such that $|\sec\phi|\leq M_\phi$.
Thus, $
    |\dot{\xi}_1|
    \leq
    M_\phi\bigl(|\Omega_1|+|\Omega_2|\bigr)+|\Omega_3|
    <\infty .
$
For the second component,
\(
    \xi_2=J_1\Omega_1+J_3c_2\Omega_3 ,
\)
and therefore
$
    |\xi_2|
    \leq
    J_1|\Omega_1|+J_3c_2|\Omega_3|
    <\infty .
$
Furthermore, from~\eqref{eq:intr_dyn},
$
    \dot{\xi}_2
    =
    -[1,0,c_2](\Omega\times J\Omega+k_\Omega J\Omega)+c_1u_t .
$
Since $\Omega$ and $u_t$ are bounded along the closed-loop trajectory, $\dot{\xi}_2$ is bounded.
Hence, both $\xi$ and $\dot{\xi}$ remain bounded.
\end{proof}
In summary, Theorem~\ref{thm:main_result} guarantees local exponential convergence to the desired path~(\ref{goal:1}) and forward invariance of the path-following manifold under rotor failure~(\ref{goal:2}). The tangential input $\nu_3$ assigns the desired motion along the path, such as point stabilization or path-speed tracking~(\ref{goal:3}), while boundedness of $\xi$ and $\dot{\xi}$ implies bounded yaw motion~(\ref{goal:4}). Thus, Problem~\ref{prob:Quasi-PF-Problem} is solved locally on the regular domain.
%%% Simulation %%%
\section{Simulation}
\label{sec:simulations}

\begin{figure}[ht]
    \centering
    \includegraphics[width=0.99\linewidth]{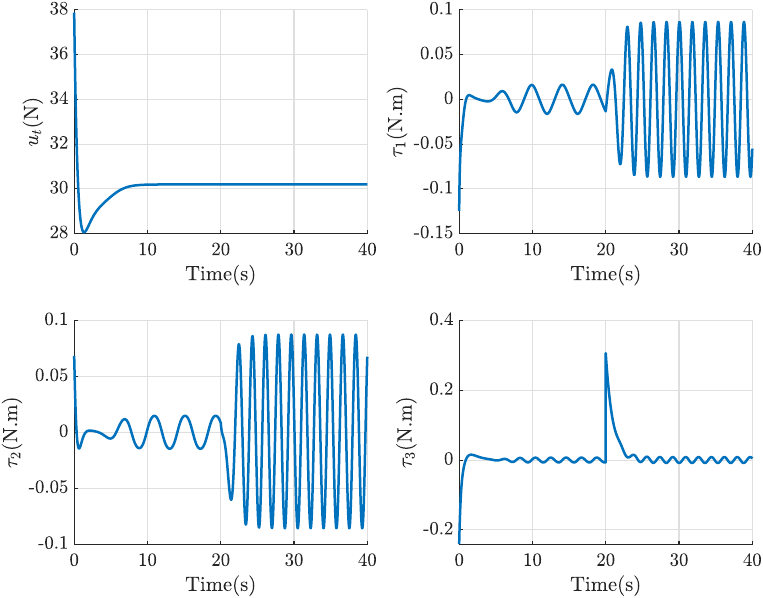}
    \vspace{- 0.5cm}
    \caption{Input time plots. Rotor failure occurs at $t=\SI{20}{\second}$.}
    \vspace{- 0.5cm}
    \label{fig:input}
\end{figure}
\begin{figure}[ht]
    \centering
    \includegraphics[width=0.99\linewidth]{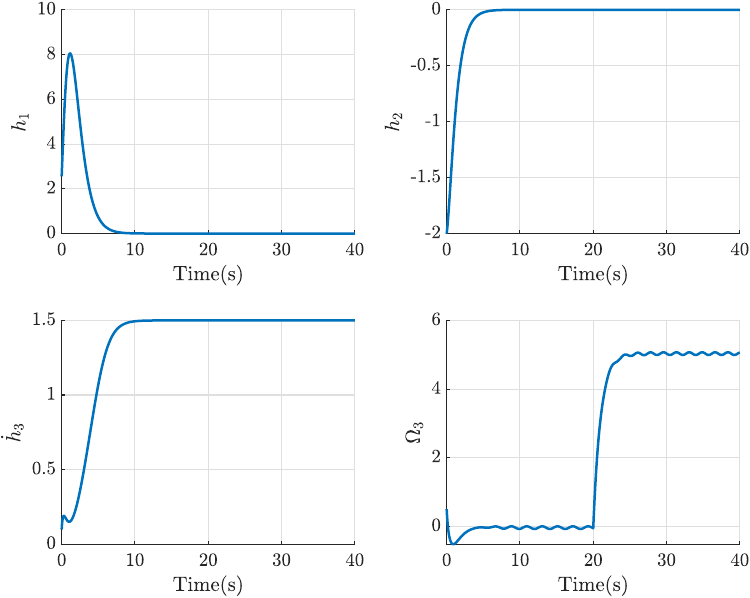}
    \vspace{- 0.5cm}
    \caption{Output time-plots along with $\Omega_3$. Rotor failure occurs at $t=\SI{20}{\second}$.}
    \vspace{- 0.5cm}
    \label{fig:output}
\end{figure}

This section illustrates the performance of the proposed QSF~\eqref{eq:ut},\eqref{eq:tau} using the Drake physics engine~\cite{drake}. The UAV is modeled in Drake, where the actual applied torque is given by~\eqref{eq:taua}. As seen in Fig.~\ref{fig:snapshots}, the UAV with four healthy rotors approaches a desired path, then experiences complete single-rotor failure mid-path. Upon failure, the fault-tolerant QSF takes over, and invariance of the path-following manifold ensures the UAV remains on the path. The UAV parameters are $m = \SI{3}{\kilogram}$, $J = \mathrm{diag}(0.03,0.03,0.06)\,\unit{\kilogram\cdot\meter\squared}$, rotor-to-CoM distance $\ell = \SI{0.2}{\meter}$, and drag-to-thrust ratio $c_d = 0.01$. The desired path is a horizontal circle of radius $d_1$ at height $d_2$, yielding output functions $h_1 = p_1^2+p_2^2-d_1^2$, $h_2 = p_3-d_2$, and $h_3 = d_1 \operatorname{atan2}(p_2,p_1)$, where $h_3$ enforces constant arc-length velocity $\dot{s}\to v_d$ (see Remark~\ref{rmk:tangential_control}). The simulation parameters are $d_1=\SI{5}{\meter}$, $d_2=\SI{-3}{\meter}$, and $v_d=\SI{1.5}{\meter\per\second}$. The control gains are chosen by placing eigenvalues at $\{-2.4,-2,-1.3,-1\}$, $\{-1,-2\}$, and $\{-1,-2,-3\}$ for the $z^1$, $z^2$, and $z^3$ subsystems, yielding $K_1=[6.24~16.76~16.22~6.7]$, $K_2=[2~3]$, and $K_3=[6~11~6]$, respectively. For the $z^3$-subsystem, $\dot{s}$ is directly controlled.  The initial conditions are $p(0)=\col{-0.7,1.75,0}\,\unit{\meter}$, $v(0)=\col{-1,2,0.5}\,\unit{\meter\per\second}$, $\eta(0)=\col{-0.1,0.2,1}\,\unit{\radian}$, and $\Omega(0)=\col{0.9,-0.1,0.5}\,\unit{\radian\per\second}$.

The simulation is shown in Fig.~\ref{fig:snapshots}. It unfolds in three stages\footnote{Code and animation at \href{https://github.com/gradslab/quasistatic-ftc}{https://github.com/gradslab/quasistatic-ftc}.}: path convergence (Figs.~\ref{fig:sub1}-\ref{fig:sub4}), nominal path-following (Figs.~\ref{fig:sub5}-\ref{fig:sub8}), and path-following under rotor failure at $t=\SI{20}{\second}$ (Figs.~\ref{fig:sub9}-\ref{fig:sub12}). In Stage 3, the invariant path-following manifold keeps the UAV on its path despite the failure. The only motion lost is yaw, which remains bounded and tunable via the damping coefficient $k_\Omega = \SI{1}{\second^{-1}}$. In the healthy four-rotor case, the controller of~\cite{LawAkh26} is used.

The input that results in this motion is shown in Fig.~\ref{fig:input}. When the failure occurs at $t=\SI{20}{\second}$, the thrust is unchanged. This is expected since the circle is horizontal and the UAV never leaves it. However, the rotor failure causes $\tau_1$ to be dependent on $\tau_2$, $\tau_3$, and the fault-tolerant QSF commands different signals for $\tau_2$, $\tau_3$.
As for the output, $h_1$, $h_2$ are shown in Fig.~\ref{fig:output}. Since this example drives $\dot{h}_3$ to a desired $v_d$, $\dot{h}_3$ is shown in Fig.~\ref{fig:output} instead of $h_3$. In addition, the body rate $\Omega_3$ is also plotted illustrating boundedness of the rotational motion about the $b_3$ axis even after rotor failure.

%%% Conclusions %%%
\section{Conclusions}
\label{sec:conclusion}
This paper presented a fault-tolerant QSF for a quadrotor UAV achieving feedback linearization without a dynamic controller in a set stabilization framework. The QSF design resulted in a static control law that renders motions on the path controlled-invariant, guaranteeing the UAV remains on its path under sudden rotor failure. Also, avoiding dynamic feedback yields considerable practical advantages. The proposed QSF guarantees exponential convergence of the UAV CoM to the desired path with bounded rotational dynamics, validated via the Drake physics engine.

\begingroup
\footnotesize
\bibliographystyle{IEEEtran}
\bibliography{references}
\endgroup

\end{document}